\documentclass[11pt]{article}
\usepackage{epsf}
\usepackage{amsmath}
\usepackage{epsfig}
\usepackage{times}
\usepackage{amssymb}
\usepackage{amsthm}
\usepackage{setspace}
\usepackage{cite}

\usepackage{algorithmic}  
\usepackage{algorithm}

\usepackage{shadow}
\usepackage{fancybox}
\usepackage{fancyhdr}

\def\w{{\bf w}}

\def\y{{\bf y}}

\def\x{{\bf x}}

\def\x{{\mathbf x}}

\def\w{{\bf w}}

\def\x{{\bf x}}
\def\y{{\bf y}}
\def\z{{\bf z}}

\def\h{{\bf h}}

\def\be{\begin{equation}}
\def\ee{\end{equation}}
\def\ba{\left[\begin{array}}
\def\ea{\end{array}\right]}

\def\t{{\bf t}}

\def\w{{\bf w}}

\def\x{{\bf x}}
\def\y{{\bf y}}
\def\z{{\bf z}}

\def\1{{\bf 1}}

\def\W{{\bf W}}

\def\0{{\bf 0}}

\def\erfinv{\mbox{erfinv}}
\def\htheta{\hat{\theta}}

\def\H{{\bf H}}
\def\X{{\bf X}}

\def\Hnorm{\H_{(norm)}}

\def\gammainc{\gamma_{inc}}
\def\gammaincpl{\gamma_{inc,+}}

\def\erfinv{\mbox{erfinv}}
\def\htheta{\hat{\theta}}

\def\Sw{S_w}

\def\hw{\bar{\h}}

\def\cweak{c_{w}}

\def\Sw{S_w}

\def\hw{\bar{\H}}
\def\cweak{c_{w}}
\def\zw{\bar{\z}}

\def\betaweak{\beta_{w}}
\def\thetaweak{\theta_{w}}
\def\hthetaweak{\hat{\theta}_{w}}

\newtheorem{theorem}{Theorem}
\newtheorem{corollary}{Corollary}

\begin{document}

\begin{singlespace}

\title {Compressed sensing of block-sparse positive vectors 
\footnote{ This work was supported in part by NSF grant \#CCF-1217857.}
}
\author{
\textsc{Mihailo Stojnic}
\\
\\
{School of Industrial Engineering}\\
{Purdue University, West Lafayette, IN 47907} \\
{e-mail: {\tt mstojnic@purdue.edu}} }
\date{}
\maketitle

\centerline{{\bf Abstract}} \vspace*{0.1in}

In this paper we revisit one of the classical problems of compressed sensing. Namely, we consider linear under-determined systems with sparse solutions. A substantial success in mathematical characterization of an $\ell_1$ optimization technique typically used for solving such systems has been achieved during the last decade. Seminal works \cite{CRT,DOnoho06CS} showed that the $\ell_1$ can recover a so-called linear sparsity (i.e. solve systems even when the solution has a sparsity linearly proportional to the length of the unknown vector). Later considerations \cite{DonohoPol,DonohoUnsigned} (as well as our own ones \cite{StojnicCSetam09,StojnicUpper10}) provided the precise characterization of this linearity. In this paper we consider the so-called structured version of the above sparsity driven problem. Namely, we view a special case of sparse solutions, the so-called block-sparse solutions. Typically one employs $\ell_2/\ell_1$-optimization as a variant of the standard $\ell_1$ to handle block-sparse case of sparse solution systems.  We considered systems with block-sparse solutions in a series of work \cite{StojnicCSetamBlock09,StojnicUpperBlock10,StojnicICASSP09block,StojnicJSTSP09} where we were able to provide precise performance characterizations if the $\ell_2/\ell_1$-optimization similar to those obtained for the standard $\ell_1$ optimization in \cite{StojnicCSetam09,StojnicUpper10}. Here we look at a similar class of systems where on top of being block-sparse the unknown vectors are also known to have components of the same sign. In this paper we slightly adjust $\ell_2/\ell_1$-optimization to account for the known signs and provide a precise performance characterization of such an adjustment.

\vspace*{0.25in} \noindent {\bf Index Terms: Compressed sensing; $\ell_2/\ell_1$ optimization; linear systems of equations; signed unknown vectors}.

\end{singlespace}

\section{Introduction}
\label{sec:back}

As usual we start by recalling on the basic mathematical definitions related to under-determined systems of linear equations. These problems are one of the mathematical cornerstones of compressed sensing (of course a great deal of work has been done in the compressed sensing; instead of reviewing it here we, for more on compressed sensing ideas, refer to the introductory papers \cite{CRT,DOnoho06CS}). Since this paper will be dealing with certain mathematical aspects of compressed sensing under-determined systems of linear equations will be its a focal point.

To insure that we are on a right mathematical track we will along these lines start with providing their an as simple as possible description. One typically starts with a systems matrix $A$ which is an $M\times N$ ($M\leq N$) dimensional matrix with real entries and then considers an $N$ dimensional vector $\tilde{\x}$ that also has real entries but on top of that no more than $K$ nonzero entries (in the rest of this paper we will call such a vector $K$-sparse). Then one forms the product of $A$ and $\tilde{\x}$ to obtain $\y$
\begin{equation}
\y=A\tilde{\x}. \label{eq:defy}
\end{equation}
Clearly, in general $\y$ is an $M$ dimensional vector with real entries. Then, for a moment one pretends that $\tilde{\x}$ is not known and poses the following inverse problem: given $A$ and $\y$ from (\ref{eq:defy}) can one then determine $\tilde{\x}$? Or in other words, can one for a given pair $A$ and $\y$ find the $k$ sparse solution of the following linear systems of equation type of problem (see, Figure \ref{fig:model})
\begin{equation}
A\x=\y. \label{eq:system}
\end{equation}
\begin{figure}[htb]
\centering
\centerline{\epsfig{figure=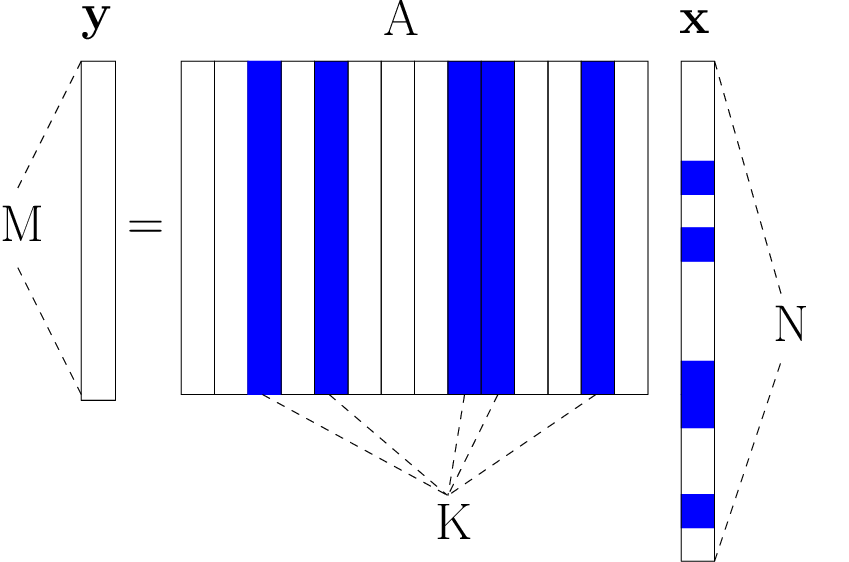,width=10.5cm,height=6cm}}
\caption{Model of a linear system; vector $\x$ is $k$-sparse}
\label{fig:model}
\end{figure}
Of course, based on (\ref{eq:defy}) such an $\x$ exists (moreover, it is an easy algebraic exercise to show that when $k<m/2$ it is in fact unique). Additionally, we will assume that there is no $\x$ in (\ref{eq:system}) that is less than $k$ sparse. One often (especially within the compressed sensing context) rewrites the problem described above (and given in (\ref{eq:system})) in the following way
\begin{eqnarray}
\mbox{min} & & \|\x\|_{0}\nonumber \\
\mbox{subject to} & & A\x=\y, \label{eq:l0}
\end{eqnarray}
where $\|\x\|_{0}$ is what is typically called $\ell_0$ norm of vector $\x$. For all practical purposes we will view $\|\x\|_{0}$ as the number that counts how many nonzero entries $\x$ has.

To make writing in the rest of the paper easier, we will assume the
so-called \emph{linear} regime, i.e. we will assume that $K=\beta N$
and that the number of equations is $M=\alpha N$ where
$\alpha$ and $\beta$ are constants independent of $N$ (more
on the non-linear regime, i.e. on the regime when $M$ is larger than
linearly proportional to $K$ can be found in e.g.
\cite{CoMu05,GiStTrVe06,GiStTrVe07}). Of course, we do mention that all of our results can easily be adapted to various nonlinear regimes as well.

Looking back at (\ref{eq:system}), clearly one can consider an exhaustive search type of solution where one would look at all subsets of $k$ columns of $A$ and then attempt to solve the resulting system. However, in the linear regime that we assumed above such an approach becomes prohibitively slow as $n$ grows. That of course led in last several decades towards a search for more clever algorithms for solving (\ref{eq:system}). Many great algorithms were developed (especially during the last decade) and many of them have even provably excellent performance measures (see, e.g. \cite{JATGomp,JAT,NeVe07,DTDSomp,NT08,DaiMil08,DonMalMon09}).

A particularly successful technique for solving (\ref{eq:system}) that will be of our interest in this paper is a linear programming relaxation of (\ref{eq:l0}), called $\ell_1$-optimization. (Variations of the standard $\ell_1$-optimization from e.g.
\cite{CWBreweighted,SChretien08,SaZh08}) as well as those from \cite{SCY08,FL08,GN03,GN04,GN07,DG08,StojnicLqThrBnds10,StojnicLiftLqThrBnds13} related to $\ell_q$-optimization, $0<q<1$
are possible as well.) Basic $\ell_1$-optimization algorithm finds $\x$ in
(\ref{eq:system}) or (\ref{eq:l0}) by solving the following $\ell_1$-norm minimization problem
\begin{eqnarray}
\mbox{min} & & \|\x\|_{1}\nonumber \\
\mbox{subject to} & & A\x=\y. \label{eq:l1}
\end{eqnarray}
Due to its popularity the literature on the use of the above algorithm is rapidly growing. Surveying it here goes way beyond the main interest of this paper and we defer doing so to a review paper. Here we just briefly mention that in seminal works \cite{CRT,DOnoho06CS} it was proven in a statistical context that for any $0<\alpha\leq 1$ there will be a $\beta$ such that $\tilde{\x}$ is the solution of (\ref{eq:l1}). \cite{DonohoPol,DonohoUnsigned} (later our own work \cite{StojnicCSetam09,StojnicUpper10} as well) for any $0<\alpha\leq 1$ determined the exact values of $\beta$ such that almost any $\tilde{\x}$ is the solution of (\ref{eq:l1}). That essentially settled a statistical performance characterization of (\ref{eq:l1}) when employed as an alternate to (\ref{eq:l0}).

The bottom line of considerations presented in \cite{StojnicCSetam09,StojnicUpper10,DonohoPol,DonohoUnsigned} is the following theorem.

\begin{theorem}(Exact threshold)
Let $A$ be an $M\times N$ matrix in (\ref{eq:system})
with i.i.d. standard normal components. Let
the unknown $\x$ in (\ref{eq:system}) be $K$-sparse. Further, let the location and signs of nonzero elements of $\x$ be arbitrarily chosen but fixed.
Let $K,M,N$ be large
and let $\alpha_w=\frac{M}{N}$ and $\beta_w=\frac{K}{N}$ be constants
independent of $M$ and $N$. Let $\erfinv$ be the inverse of the standard error function associated with zero-mean unit variance Gaussian random variable.  Further, let $\alpha_w$ and $\beta_w$ satisfy the following:

\noindent \underline{\underline{\textbf{Fundamental characterization of the $\ell_1$ performance:}}}

\begin{center}
\shadowbox{$
(1-\beta_w)\frac{\sqrt{\frac{2}{\pi}}e^{-(\erfinv(\frac{1-\alpha_w}{1-\beta_w}))^2}}{\alpha_w}-\sqrt{2}\erfinv (\frac{1-\alpha_w}{1-\beta_w})=0.
$}
-\vspace{-.5in}\begin{equation}
\label{eq:thmweaktheta2}
\end{equation}
\end{center}

Then:
\begin{enumerate}
\item If $\alpha>\alpha_w$ then with overwhelming probability the solution of (\ref{eq:l1}) is the $k$-sparse $\x$ from (\ref{eq:system}).
\item If $\alpha<\alpha_w$ then with overwhelming probability there will be a $K$-sparse $\x$ (from a set of $\x$'s with fixed locations and signs of nonzero components) that satisfies (\ref{eq:system}) and is \textbf{not} the solution of (\ref{eq:l1}).
    \end{enumerate}
then with overwhelming probability there will be a $K$-sparse $\x$ (from a set of $\x$'s with fixed locations and signs of nonzero components) that satisfies (\ref{eq:system}) and is \textbf{not} the solution of (\ref{eq:l1}).
\label{thm:thmweakthr}
\end{theorem}
\begin{proof}
The first part was established in \cite{StojnicCSetam09} and the second one was established in \cite{StojnicUpper10}. An alternative way of establishing the same set of results was also presented in \cite{StojnicEquiv10}. Of course, similar results were obtained initially in \cite{DonohoPol,DonohoUnsigned}.
\end{proof}

As mentioned above, the above theorem (as well as corresponding results obtained earlier in \cite{DonohoPol,DonohoUnsigned})) essentially settled typical behavior of $\ell_1$ optimization when used for solving (\ref{eq:system}) or (\ref{eq:l0}). In this paper we will look at a problem similar to the one from (\ref{eq:l0}) (or (\ref{eq:system})). Namely, we will view problem from (\ref{eq:system})) within the following framework: we will assume that $\tilde{\x}$ is not only sparse but also what is called block-sparse. Such an assumption can then be incorporated in the recovery algorithm as well. Before proceeding further with the presentation of such an algorithm we briefly sketch how the rest of the paper will be organized. In Section
\ref{sec:blsppos} we first introduce the block-sparse vectors and their a special variant that we will call \emph{positive} block-sparse vectors. In Section \ref{sec:analblsppos} we then present a performance analysis of an algorithm that can be used for solving linear under-determined systems known to have \emph{positive} block-sparse solutions. Finally, in Sections \ref{sec:conc} we discuss obtained results and provide a few conclusions.

\section{Block-sparse positive vectors}
\label{sec:blsppos}

What we described in the previous section assumes solving an under-determined system of linear equations with a standard restriction that the solution vector is sparse. Sometimes one may however encounter
applications when the unknown $\x$ in addition to being sparse has a
certain structure as well. The so-called block-sparse vectors are such a type of vectors and will be the main subject of this paper. These vectors and their potential applications and recovery algorithms were
investigated to a great detail in a series of recent references (see e.g. \cite{EldBol09,EKB09,SPH,FHicassp,EMsub,BCDH08,StojnicICASSP09block,StojnicJSTSP09,GaZhMa09,CeInHeBa09}). A related problem
of recovering jointly sparse vectors and its applications were also
considered to a great detail in e.g. \cite{ZeGoAd09,ZeWaSeGoAd08,TGS05,BWDSB05,CH06,CREKD,MEldar,Temlyakov04,VPH,BerFri09,EldRau09,BluDav09,NegWai09} and many
references therein. While various other structures as well as their applications gained significant interest over last few years we here refrain from describing them into fine details and instead refer to nice work of e.g.
\cite{KDXH09,XKAH09,KXAH09,RFPrank,StojnicISIT2010binary}. Since we will be interested in characterizing mathematical properties of solving linear systems that are similar to many of those mentioned above we just state here in brief that from a mathematical point of view in all these cases one attempts to improve the
recoverability potential of the standard algorithms (which are typically similar to the one described in the
previous section) by incorporating the knowledge of the unknown vector
structure.

To get things started we first introduce the block-sparse vectors.
The subsequent exposition will also be somewhat less cumbersome if we assume that
integers $N$ and $d$ are chosen such that $n=\frac{N}{d}$ is an
integer and it represents the total number of blocks that $\x$
consists of. Clearly $d$ is the length of each block. Furthermore,
we will assume that $m=\frac{M}{d}$ is an integer as well and that
$\X_i=\x_{(i-1)d+1:id}, 1\leq i\leq n$, are the $n$ blocks of $\x$
(see Figure \ref{fig:blspmodel}).
\begin{figure}[htb]
\centering
\centerline{\epsfig{figure=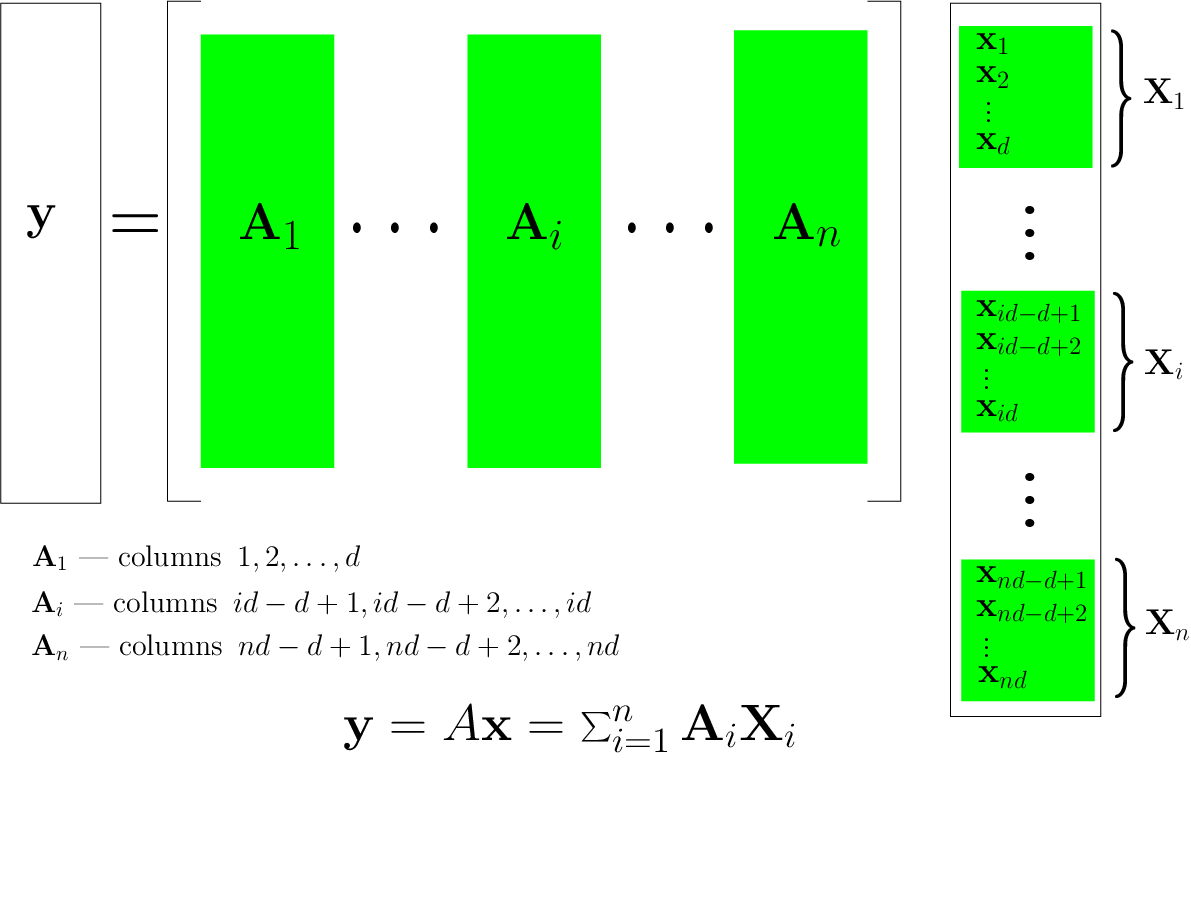,width=12cm,height=10cm}}
\vspace{-0.55in} \caption{Block-sparse model} \label{fig:blspmodel}
\end{figure}
Then we will call any signal $\x$ k-block-sparse if its at most
$k=\frac{K}{d}$ blocks $\X_i$ are non-zero (non-zero block is a block
that is not a zero block; zero block is a block that has all
elements equal to zero). Since $k$-block-sparse signals are
$K$-sparse one could then use (\ref{eq:l1}) to recover the solution
of (\ref{eq:system}). While this is possible, it clearly uses the
block structure of $\x$ in no way. There are of course many ways how one can attempt to exploit the block-sparse structure. Below we just briefly mentioned a few of them.

A few approaches from a vast literature cited above have recently attracted significant amount of attention. The first thing one can think of when facing the block-structured unknown vectors is how to extend results known in the non-block (i.e. standard) case. In \cite{TGS05} the standard OMP (orthogonal matching pursuit) was generalized so that it can handle the jointly-sparse vectors more efficiently and improvements over the standard OMP were demonstrated. In \cite{BCDH08,EMsub} algorithms similar to the one from this paper were considered. It was explicitly shown through the block-RIP (block-restricted isometry property) type of analysis (which essentially extends to the block case the concepts introduced in \cite{CRT} for the non-block scenario) that one can achieve improvements in recoverable thresholds compared to the non-block case. Also, important results were obtained in \cite{EldRau09} where it was shown (also through the block-RIP type of analysis) that if one considers average case recovery of jointly-sparse signals the improvements in recoverable thresholds over the standard non-block signals are possible (of course, trivially, jointly-sparse recovery offers no improvement over the standard non-block scenario in the worst case). All these results provided a rather substantial basis for belief that the block-sparse recovery can provably be significantly more successful than the standard non-block one.

To exploit the block structure of
$\x$ in \cite{SPH} the following polynomial-time algorithm (essentially a
combination of $\ell_2$ and $\ell_1$ optimizations) was considered (see also e.g. \cite{ZeWaSeGoAd08,ZeGoAd09,BCDH08,EKB09,BerFri09})
\begin{eqnarray}
\mbox{min} & & \sum_{i=1}^{n}\|\x_{(i-1)d+1:id}\|_2\nonumber \\
\mbox{subject to} & & A\x=\y. \label{eq:l2l1}\vspace{-.1in}
\end{eqnarray}
Extensive simulations in \cite{SPH} demonstrated that as $d$ grows
the algorithm in (\ref{eq:l2l1}) significantly outperforms the
standard $\ell_1$. The following was shown in \cite{SPH} as well:
let $A$ be an $M\times N$ matrix with a basis of null-space
comprised of i.i.d. Gaussian elements; if
$\alpha=\frac{M}{N}\rightarrow 1$ then there is a constant $d$ such
that all $k$-block-sparse signals $\x$ with sparsity $K\leq \beta N,
\beta\rightarrow \frac{1}{2}$, can be recovered with overwhelming
probability by solving (\ref{eq:l2l1}).
The precise relation between
$d$ and how fast $\alpha\longrightarrow 1$ and $\beta\longrightarrow
\frac{1}{2}$ was quantified in \cite{SPH} as well. In \cite{StojnicICASSP09block,StojnicJSTSP09} we extended the results from
\cite{SPH} and obtained the values of the recoverable block-sparsity for any
$\alpha$, i.e. for $0\leq \alpha \leq 1$. More precisely, for any
given constant $0\leq \alpha \leq 1$ we in \cite{StojnicICASSP09block,StojnicJSTSP09} determined a constant
$\beta=\frac{K}{N}$ such that for a sufficiently large $d$ (\ref{eq:l2l1})
with overwhelming probability
recovers any $k$-block-sparse signal with sparsity less then $K$.
(Under overwhelming probability we in this paper assume
a probability that is no more than a number exponentially decaying in $N$ away from $1$.)


Clearly, for any given constant $\alpha\leq 1$ there is a maximum
allowable value of $\beta$ such that for \emph{any} given $k$-sparse $\x$ in (\ref{eq:system}) the solution of (\ref{eq:l2l1})
is with overwhelming probability exactly that given $k$-sparse $\x$. This value of
$\beta$ is typically referred to as the \emph{strong threshold} (see
\cite{DonohoPol,StojnicCSetamBlock09}). Similarly, for any given constant
$\alpha\leq 1$ and \emph{any} given $\x$ with a given fixed location and a given fixed directions of non-zero blocks
there will be a maximum allowable value of $\beta$ such that
(\ref{eq:l2l1}) finds that given $\x$ in (\ref{eq:system}) with overwhelming
probability. We will refer to this maximum allowable value of
$\beta$ as the \emph{weak threshold} and will denote it by $\beta_{w}$ (see, e.g. \cite{StojnicICASSP09,StojnicCSetam09}).

While \cite{StojnicICASSP09block,StojnicJSTSP09} provided fairly sharp strong threshold values they had done so in a somewhat asymptotic sense. Namely, the analysis presented in \cite{StojnicICASSP09block,StojnicJSTSP09} assumed fairly large values of block-length $d$. As such the analysis in \cite{StojnicICASSP09block,StojnicJSTSP09} then provided an ultimate performance limit of $\ell_2/\ell_1$-optimization rather than its performance characterization as a function of a particular fixed block-length.

In our own work \cite{StojnicCSetamBlock09} we extended the results of \cite{StojnicICASSP09block,StojnicJSTSP09} and provided a novel probabilistic framework for performance characterization of (\ref{eq:l2l1}) through which we were finally able to view block-length as a parameter of the system (the heart of the framework was actually introduced in \cite{StojnicCSetam09} and it seemed rather powerful; in fact, we afterwards found hardly any sparse type of problem that the framework was not able to handle with an almost impeccable precision). Using the framework we obtained lower bounds on $\beta_w$. These lower bounds were in an excellent numerical agreement with the values obtained for $\beta_w$ through numerical simulations. In a followup \cite{StojnicUpperBlock10} we then showed that the lower bounds on $\beta_w$ obtained in \cite{StojnicCSetamBlock09} are actually exact.

The following theorem essentially summarizes the results obtained in \cite{StojnicCSetamBlock09,StojnicUpperBlock10} and effectively establishes for any $0<\alpha\leq 1$ the exact value of $\beta_w$ for which (\ref{eq:l2l1}) finds the $k$-block-sparse $\x$ from (\ref{eq:system}).

\begin{theorem}(\cite{StojnicCSetamBlock09,StojnicUpperBlock10} Exact weak threshold; block-sparse $\x$)
Let $A$ be an $M\times N$ matrix in (\ref{eq:system})
with i.i.d. standard normal components. Let
the unknown $\x$ in (\ref{eq:system}) be $k$-block-sparse with the length of its blocks $d$. Further, let the location and the directions of nonzero blocks of $\x$ be arbitrarily chosen but fixed.
Let $k,m,n$ be large
and let $\alpha=\frac{m}{n}$ and $\betaweak=\frac{k}{n}$ be constants
independent of $m$ and $n$. Let $\gammainc(\cdot,\cdot)$ and $\gammainc^{-1}(\cdot,\cdot)$ be the incomplete gamma function and its inverse, respectively. Further,
let all $\epsilon$'s below be arbitrarily small constants.

\begin{enumerate}
\item Let $\hthetaweak$, ($\betaweak\leq \hthetaweak\leq 1$) be the solution of
\begin{equation}
(1-\epsilon_1^{(c)})(1-\betaweak)\frac{\frac{\sqrt{2}\Gamma(\frac{d+1}{2})}{\Gamma(\frac{d}{2})}
\left (1-\gammainc(\gammainc^{-1}(\frac{1-\thetaweak}{1-\betaweak},\frac{d}{2}),\frac{d+1}{2})\right )}{\thetaweak}-\sqrt{2\gammainc^{-1}(\frac{(1+\epsilon_1^{(c)})(1-\thetaweak)}{1-\betaweak},\frac{d}{2})}=0
.\label{eq:thmweaktheta}
\end{equation}
If $\alpha$ and $\betaweak$ further satisfy
\begin{multline}
\alpha d>(1-\betaweak)\frac{2\Gamma(\frac{d+2}{2})}{\Gamma(\frac{d}{2})}
\left (1-\gammainc(\gammainc^{-1}(\frac{1-\hthetaweak}{1-\betaweak},\frac{d}{2}),\frac{d+2}{2})\right )
+\betaweak d\\-\frac{\left ((1-\betaweak)\frac{\sqrt{2}\Gamma(\frac{d+1}{2})}{\Gamma(\frac{d}{2})}
(1-\gammainc(\gammainc^{-1}(\frac{1-\hthetaweak}{1-\betaweak},\frac{d}{2}),\frac{d+1}{2}))\right ) ^2}{\hthetaweak}\label{eq:thmweakalpha}
\end{multline}
then with overwhelming probability the solution of (\ref{eq:l2l1}) is the $k$-block-sparse $\x$ from (\ref{eq:system}).

\item Let $\htheta_w$, ($\beta_w\leq \htheta_w\leq 1$) be the solution of
\begin{equation}
(1+\epsilon_2^{(c)})(1-\betaweak)\frac{\frac{\sqrt{2}\Gamma(\frac{d+1}{2})}{\Gamma(\frac{d}{2})}
\left (1-\gammainc(\gammainc^{-1}(\frac{1-\thetaweak}{1-\betaweak},\frac{d}{2}),\frac{d+1}{2})\right )}{\thetaweak}-\sqrt{2\gammainc^{-1}(\frac{(1-\epsilon_2^{(c)})(1-\thetaweak)}{1-\betaweak},\frac{d}{2})}=0
.\label{eq:thmweaktheta1}
\end{equation}
If $\alpha$ and $\betaweak$ further satisfy
\begin{multline}
\alpha d<\frac{1}{(1+\epsilon_1^{(m)})^2}((1-\epsilon_1^{(g)}) (1-\betaweak)\frac{2\Gamma(\frac{d+2}{2})}{\Gamma(\frac{d}{2})}
\left (1-\gammainc(\gammainc^{-1}(\frac{1-\hthetaweak}{1-\betaweak},\frac{d}{2}),\frac{d+2}{2})\right )
+\betaweak d\\-\frac{\left ((1-\betaweak)\frac{\sqrt{2}\Gamma(\frac{d+1}{2})}{\Gamma(\frac{d}{2})}
(1-\gammainc(\gammainc^{-1}(\frac{1-\hthetaweak}{1-\betaweak},\frac{d}{2}),\frac{d+1}{2}))\right ) ^2}{\hthetaweak (1+\epsilon_3^{(g)})^{-2}} )\label{eq:thmweakalpha1}
\end{multline}
then with overwhelming probability there will be a $k$-block-sparse $\x$ (from a set of $\x$'s with fixed locations and directions of nonzero blocks) that satisfies (\ref{eq:system}) and is \textbf{not} the solution of (\ref{eq:l2l1}).
\end{enumerate}
\label{thm:thmweakthrblock}
\end{theorem}
\begin{proof}
The first part was established in \cite{StojnicCSetamBlock09}. The second part was established in \cite{StojnicUpperBlock10}.
\end{proof}

As we mentioned above, the block-sparse structure is clearly not the only one that can be imposed on the unknown $\x$. Among the most typical ones usually considered in parallel with the standard sparse scenario is what is called the case of signed (say positive-nonnegative/negative-nonpositive) vectors $\x$. Such vectors can find applications in scenarios when the physics of the problem does not allow for different signs of the components of $\x$. Also, they have been of interest for a long time from a purely mathematical point of view (see, e.g. \cite{DonohoSigned,DT,StojnicCSetam09,StojnicUpper10,StojnicICASSP09,StojnicICASSP10var}). In scenarios when it is known that all components of $\x$ are of the same sign (say positive) one typically replaces (\ref{eq:l1}) with the following
\begin{eqnarray}
\mbox{min} & & \|\x\|_{1}\nonumber \\
\mbox{subject to} & & A\x=\y\nonumber \\
& & \x_i\geq 0. \label{eq:l1non}
\end{eqnarray}
As mentioned above, in \cite{StojnicCSetam09,StojnicUpper10} we showed how one characterize the performance of (\ref{eq:l1non}). In this paper we look at the above mentioned block-sparse vectors that are also signed (say positive). In such scenarios one can use the following modification of (\ref{eq:l2l1})
\begin{eqnarray}
\mbox{min} & & \sum_{i=1}^{n}\|\x_{(i-1)d+1:id}\|_2\nonumber \\
\mbox{subject to} & & A\x=\y\nonumber \\
& & \x_i\geq 0. \label{eq:l2l1non}
\end{eqnarray}
Of course, (\ref{eq:l2l1non}) relates to (\ref{eq:l2l1}) in the same way (\ref{eq:l1non}) relates to (\ref{eq:l1}). As analysis of \cite{StojnicCSetam09,StojnicUpper10} showed one can achieve a higher recoverable sparsity when components of $\x$ are a priori known to be say positive (i.e. of the same sign). This of course is not surprising since one would naturally expect that the more available information about unknown $\x$ the easier its recovery.
Along the same lines, one can then expect that a similar phenomenon should occur when one is dealing with the block-sparse signals. Among other things, the analysis that we will present below will confirm such an expectation. Of course, the key feature of what will present below will be a precise performance analysis of (\ref{eq:l2l1non}) when used for recovery of positive block-sparse vectors $\x$ in (\ref{eq:system}).

\section{Performance analysis of (\ref{eq:l2l1non})}
\label{sec:analblsppos}

In this section we will attempt to obtain the results qualitatively similar to those presented in Theorems \ref{thm:thmweakthr} and \ref{thm:thmweakthrblock}. Of course, the results presented in Theorems \ref{thm:thmweakthr} and \ref{thm:thmweakthrblock} are related to performance of (\ref{eq:l1}) and (\ref{eq:l2l1}), respectively, whereas here we will try to create their an analogue that relates to (\ref{eq:l2l1non}). As mentioned earlier, the results presented in Theorems \ref{thm:thmweakthr} and \ref{thm:thmweakthrblock} were obtained in a series of work \cite{StojnicICASSP09,StojnicCSetam09,StojnicUpper10,StojnicCSetamBlock09,StojnicUpperBlock10}. Below, we adapt some of these results so that they can handle the problems of interest here. In doing so, we will in this and all subsequent sections assume a substantial level of familiarity with many of the well-known results that relate to the performance characterization of (\ref{eq:l1}) and (\ref{eq:l2l1}) (we will fairly often recall on many results/definitions that we established in \cite{StojnicICASSP09,StojnicCSetam09,StojnicUpper10,StojnicCSetamBlock09,StojnicUpperBlock10}).

Before proceeding further with a detail presentation we briefly sketch what specifically we will be interested in showing below. Namely, using the analysis of \cite{StojnicCSetam09,StojnicUpper10,StojnicCSetamBlock09,StojnicUpperBlock10} mentioned earlier, for a specific group of randomly generated matrices $A$, one can determine values $\beta_w$ for the entire range of $\alpha$, i.e. for $0\leq \alpha\leq 1$, where $\beta_{w}$ is the
maximum allowable value of $\beta$ such that
(\ref{eq:l2l1non}) finds the positive $k$-block-sparse solution of (\ref{eq:system}) with overwhelming
probability for \emph{any} $k$-block-sparse $\x$ with given a fixed location and a fixed combination of directions of nonzero blocks, and \emph{a priori} known to have non-negative components. (As mentioned earlier and discussed to great extent in \cite{StojnicICASSP10var,StojnicICASSP09,StojnicCSetam09,StojnicUpper10,StojnicCSetamBlock09,StojnicUpperBlock10}, this value of $\beta_w$ is often referred to as the \emph{weak} threshold.)

We are now ready to start the analysis. We begin by recalling on a theorem from \cite{StojnicUpperBlock10} that provides a characterization as to when the solution of (\ref{eq:l2l1non}) is $\tilde{\x}$, i.e. the positive $k$-block-sparse solution of (\ref{eq:system}) or (\ref{eq:l0}). Since the analysis will clearly be irrelevant with respect to what particular location and what particular combination of directions of nonzero
blocks are chosen, we can for the simplicity of the exposition and without a loss of generality assume that the
blocks $\X_1,\X_2,\dots,\X_{n-k}$ of $\x$ are equal to zero and the blocks $\X_{n-k+1},\X_{n-k+2},\dots,\X_n$ of $\X$ have fixed
directions. Also, as mentioned above, we will assume that all components of $\x$ are non-negative, i.e. $\x_i\geq 0, 0\leq i\leq N$. Under these assumptions we have the following theorem (similar characterizations adopted in various contexts can be found in
\cite{DH01,XHapp,SPH,StojnicICASSP09,GN03,StojnicCSetamBlock09,StojnicUpperBlock10,StojnicICASSP10block,StojnicICASSP09block}).
\begin{theorem}(\cite{StojnicCSetamBlock09,StojnicUpperBlock10,StojnicICASSP10block} Nonzero part of $\x$ has fixed directions and location)
Assume that an $dm\times dn$ matrix $A$ is given. Let $\x$
be a positive $k$-block-sparse vector from $R^{dn}$. Also let $\X_1=\X_2=\dots=\X_{n-k}=0$ and let the directions of vectors $\X_{n-k+1},\X_{n-k+2},\dots,\X_n$ be fixed.
Further, assume that $\y=A\x=\sum_{i=1}^n A_i\X_i$ and that $\w$ is
an $dn\times 1$ vector with blocks $\W_i,i=1,\dots,n$, defined in a way analogous to the definition of blocks $\X_i$. If
\begin{equation}
(\forall \w\in \textbf{R}^{dn} | A\w=0, \w_i\geq 0, 1\leq i\leq d(n-k)) \quad  -\sum_{i=n-k+1}^n \frac{\X_i^T\W_i}{\|\X_i\|_2}<\sum_{i=1}^{n-k}\|\W_{i}\|_2.
\label{eq:thmeqgenweak1block}
\end{equation}
then the solution of (\ref{eq:l2l1}) is $\x$. Moreover, if
\begin{equation}
(\exists \w\in \textbf{R}^{dn} | A\w=0, \w_i\geq 0, 1\leq i\leq d(n-k)) \quad  -\sum_{i=n-k+1}^n \frac{\X_i^T\W_i}{\|\X_i\|_2}>\sum_{i=1}^{n-k}\|\W_{i}\|_2.
\label{eq:thmeqgenweak2block}
\end{equation}
then there will be a positive $k$-block-sparse $\x$ from the above defined set that satisfies (\ref{eq:system}) and is not the solution of (\ref{eq:l2l1non}).
\label{thm:thmgenweakblock}
\end{theorem}
\begin{proof}
The first part follows directly from Corollary $2$ in \cite{StojnicCSetamBlock09}. The second part was considered in \cite{StojnicUpperBlock10} and it follows by combining (adjusting to the block case) the first part and the ideas of the second part of Theorem $1$ (Theorem $4$) in \cite{StojnicUpper10}.
\end{proof}
Having matrix $A$ such that
(\ref{eq:thmeqgenweak1block}) holds would be enough for solutions of (\ref{eq:l2l1non}) and (\ref{eq:l0}) (or (\ref{eq:system})) to coincide. If one
assumes that $m$ and $k$ are proportional to $n$ (the case of our
interest in this paper) then the construction of the deterministic
matrices $A$ that would satisfy
(\ref{eq:thmeqgenweak1block}) is not an easy task (in fact, one may say that together with the ones that correspond to the standard $\ell_1$ it is one of the most fundamental open problems in the area of theoretical compressed sensing). However, turning to
random matrices significantly simplifies things. That is the route that will pursuit below. In fact to be a bit more specific, we will assume that the elements of matrix $A$ are i.i.d. standard normal random variables. All results that we will present below will hold for many other types of randomness (we will comment on this in more detail in Section \ref{sec:conc}). However, to make the presentation as smooth as possible we assume the standard Gaussian scenario.

We then follow the strategy of \cite{StojnicCSetam09,StojnicCSetamBlock09}. To that end we will make use of the following theorem:
\begin{theorem}(\cite{Gordon88} Escape through a mesh)
\label{thm:Gordonmesh} Let $S$ be a subset of the unit Euclidean
sphere $S^{dn-1}$ in $R^{dn}$. Let $Y$ be a random
$d(n-m)$-dimensional subspace of $R^{dn}$, distributed uniformly in
the Grassmanian with respect to the Haar measure. Let
\begin{equation}
w(S)=E\sup_{\w\in S} (\h^T\w) \label{eq:widthdef}
\end{equation}
where $\h$ is a random column vector in $R^{dn}$ with i.i.d. ${\cal
N}(0,1)$ components. Assume that
$w(S)<\left ( \sqrt{dm}-\frac{1}{4\sqrt{dm}}\right )$. Then
\begin{equation}
P(Y\cap S= \emptyset )>1-3.5e^{-\frac{\left (\sqrt{dm}-\frac{1}{4\sqrt{dm}}-w(S) \right ) ^2}{18}}.
\label{eq:thmesh}
\end{equation}
\end{theorem}

As mentioned above, to make use of Theorem \ref{thm:Gordonmesh} we follow the strategy presented in \cite{StojnicCSetam09,StojnicCSetamBlock09}. We start by defining a set $Sw'$
\begin{equation}
\Sw'=\{\w\in S^{dn-1}| \quad -\sum_{i=n-k+1}^n \frac{\X_i^T\W_i}{\|\X_i\|_2}\geq\sum_{i=1}^{n-k}\|\W_{i}\|_2\}\label{eq:defSwpr}
\end{equation}
and
\begin{equation}
w(\Sw')=E\sup_{\w\in \Sw'} (\h^T\w) \label{eq:widthdefSwpr}
\end{equation}
where as earlier $\h$ is a random column vector in $R^{dn}$ with i.i.d. ${\cal
N}(0,1)$ components and $S^{dn-1}$ is the unit $dn$-dimensional sphere. Let $\H_i = (\h_{(i-1)d + 1}, \h_{(i-1)d + 2}, \ldots, \h_{id})^T$,
$i = 1, 2, \ldots, n$ and let $\Theta_i$ be the orthogonal matrices such that $\X_i^T\Theta_i=(\|\X_i\|_2,0,\dots,0),n-k+1\leq i\leq n$. Set
\begin{equation}
\Sw=\{\w\in S^{dn-1}| \quad -\sum_{i=n-k+1}^n \w_{(i-1)d+1}\geq\sum_{i=1}^{n-k}\|\W_{i}\|_2\}\label{eq:defSw}
\end{equation}
and
\begin{equation}
w(\Sw)=E\sup_{\w\in \Sw} (\h^T\w). \label{eq:widthdefSw}
\end{equation}
Since $\H_i^T$ and $\H_i^T\Theta_i$ have the same distribution we have $w(\Sw)=w(\Sw')$.  The strategy of \cite{StojnicCSetam09,StojnicCSetamBlock09} assumes roughly the following: if $w(\Sw)< \sqrt{dm}-\frac{1}{4\sqrt{dm}}$ is positive with overwhelming probability for certain combination of $k$, $m$, and $n$ then for $\alpha=\frac{m}{n}$ one has a lower bound $\beta_{w}=\frac{k}{n}$ on the true value of the \emph{weak} threshold with overwhelming probability (we recall that as usual under overwhelming probability we of course assume a probability that is no more than a number exponentially decaying in $n$ away from $1$). The above basically means that if one can handle $w(\Sw)$ then, when $n$ is large one can, roughly speaking, use the condition $w(\Sw)< \sqrt{m}$ to obtain an attainable lower bound $\beta_{w}$ for any given $0<\alpha\leq 1$.

To that end we then look at
\begin{equation}
w(S_{w}^{(p)})=E\max_{\w\in S_{w}} (\h^T\w), \label{eq:widthdefswp}
\end{equation}
where to make writing simpler we have replaced the $\sup$ from (\ref{eq:widthdefSw}) with a $\max$. Let
\begin{eqnarray}
\H_i^* & = & (\h_{(i-1)d + 2}, \h_{(i-1)d + 3}, \ldots, \h_{id})^T,i = n-k+1, 2, \ldots, n \nonumber \\
\W_i^* & = & (\w_{(i-1)d + 2}, \w_{(i-1)d + 3}, \ldots, \w_{id})^T,i = n-k+1, 2, \ldots, n.\label{eq:defHstWst}
\end{eqnarray}
Also set
\begin{equation*}
\H_i^+ = (\max(\h_{(i-1)d + 1},0), \max(\h_{(i-1)d + 2},0), \ldots, \max(\h_{id},0))^T,1\leq i\leq n-k.
\end{equation*}
Following further what was done in \cite{StojnicCSetam09,StojnicCSetamBlock09} one then can write
\begin{equation}
w(\Sw)=E\max_{\w\in \Sw} (\h^T\w)=\max_{\w\in \Sw} (\sum_{i=n-k+1}^n \h_{(i-1)d+1}\w_{(i-1)d+1}+
\sum_{i=n-k+1}^n\|\H_i^*\|_2\|\W_i^*\|_2+\sum_{i=1}^{n-k}\|\H_i^+\|_2\|\W_i\|_2).\label{eq:workww0}
\end{equation}
Set $\Hnorm^{(n-k,+)}=(\|\H_1^+\|_2,\|\H_2^+\|_2,\dots,\|\H_{n-k}^+\|_2)$ and let $|\Hnorm^{(n-k,+)}|_{(i)}$ be the $i$-th element in the sequence of elements of
$\Hnorm^{(n-k,+)}$ sorted in increasing order. Set
\begin{multline}
\hw^+=(|\Hnorm^{(n-k,+)}|_{(1)},|\Hnorm^{(n-k,+)}|_{(2)},\dots,|\Hnorm^{(n-k,+)}|_{(n-k)},
-\h_{(n-k+1)d+1},-\h_{(n-k+2)d+1},\dots,-\h_{(n-1)d+1},\\
\|\H_{n-k+1}^*\|_2,\|\H_{n-k+2}^*\|_2,\dots,\|\H_{n}^*\|_2)^T.\label{eq:defhweak}
\end{multline}
Let $\bar{\y}=(\y_1,\y_2,\dots,\y_{n+k})^T\in R^{n+k}$. Then one can simplify (\ref{eq:workww0}) in the following way
\begin{eqnarray}
w(\Sw) = \max_{\bar{\t}\in R^{n+k}} & &  \sum_{i=1}^{n+k} \hw_i^+ \bar{\t}_i\nonumber \\
\mbox{subject to} &  & \bar{\t}_i\geq 0, 0\leq i\leq n-k,n+1\leq i\leq n+k\nonumber \\
& & \sum_{i=n-k+1}^n\bar{\t}_i\geq \sum_{i=1}^{n-k} \bar{\t}_i \nonumber \\
& & \sum_{i=1}^{n+k}\bar{\t}_i^2\leq 1\label{eq:workww2}
\end{eqnarray}
where $\hw_i^+$ is the $i$-th element of $\hw^+$. Let $\zw\in R^{n+k}$ be a vector such that $\zw_i=1,1\leq i\leq n-k$, $\zw_i=-1,n-k+1\leq i\leq n$, and $\zw_i=0,n+1\leq i\leq n+k$.

Following step by step the derivation in \cite{StojnicCSetam09,StojnicCSetamBlock09} one has, based on the Lagrange duality theory, that there is a $\cweak=(1-\theta_w)n\leq (n-k)$ such that
\begin{eqnarray}
\hspace{-.5in}\lim_{n\rightarrow \infty} \frac{w(\Sw)}{\sqrt{n}}=\lim_{n\rightarrow \infty} \frac{E\max_{\w\in S_{w}} (\h^T\w)}{\sqrt{n}}
& \approxeq &
\sqrt{\lim_{n\rightarrow \infty} \frac{E\sum_{i=\cweak+1}^{n+k}(\hw_i^+)^2}{n}-\frac{(\lim_{n\rightarrow \infty} \frac{E((\hw^+)^T\z)-E\sum_{i=1}^{\cweak}\hw_i^+}{n})^2}{1-\lim_{n\rightarrow \infty}\frac{\cweak}{n}}}\nonumber \\
& = &
\sqrt{\lim_{n\rightarrow \infty} \frac{E\sum_{i=\cweak+1}^{n+k}(\hw_i^+)^2}{n}-\frac{(\lim_{n\rightarrow \infty} \frac{E((\hw^+)^T\z)-E\sum_{i=1}^{\cweak}\hw_i^+}{n})^2}{\thetaweak}}.\nonumber \\ \label{eq:wwp}
\end{eqnarray}
where we recall that $\hw_i^+$ is the $i$-th element of vector $\hw^+$. Moreover, \cite{StojnicCSetam09} also establishes the way to determine a critical $\cweak$. Roughly speaking it establishes the following identity
\begin{equation}
\frac{(\lim_{n\rightarrow \infty} \frac{E((\hw^+)^T\z)-E\sum_{i=1}^{\cweak}\hw_i^+}{n})}{1-\lim_{n\rightarrow \infty}\frac{\cweak}{n}}
=\frac{(\lim_{n\rightarrow \infty} \frac{E((\hw^+)^T\z)-E\sum_{i=1}^{\cweak}\hw_i^+}{n})}{\thetaweak}\approxeq \lim_{n\rightarrow \infty} E\hw_{\cweak}^+.\label{eq:condcwp}
\end{equation}
To make the above results operational one would have to estimate the expectations that they contain. \cite{StojnicCSetam09,StojnicCSetamBlock09} established a technique powerful enough to do so. However, differently from \cite{StojnicCSetam09,StojnicCSetamBlock09} one has to be fairly careful when it comes to the distributions of the underlying random quantities. While the corresponding ones in \cite{StojnicCSetam09,StojnicCSetamBlock09} were relatively simple the ones that we face here are a bit harder to analytically quantify. We, hence present these considerations in a separate section below.

\subsection{Explicit characterization of $\lim_{n\rightarrow \infty} \frac{w(\Sw)}{\sqrt{n}}$}
\label{sec:expcharwsw}

We separately characterize all of the quantities needed for characterization of $\lim_{n\rightarrow \infty} \frac{w(\Sw)}{\sqrt{n}}$.

\subsubsection{Explicit characterization of $\lim_{n\rightarrow\infty}E\hw_{\cweak}^+$}
\label{sec:invthetaweak}

We start by characterizing $\hw_{cweak}^+$. To that end we define a random variable $\chi_d^+$ in the following way
\begin{equation}
(\chi_d^+)^2=\sum_{i=1}^d\max(\h_i,0)^2.\label{eq:defxidpl}
\end{equation}
Consider the following function $\gammaincpl(\cdot,d)$
\begin{equation}
\gammaincpl(\cdot,d)=\sum_{d_{ind}=0}^d\frac{\binom{d}{d_{ind}}}{2^d}\gammainc(\cdot,\frac{d_{ind}}{2}),\label{eq:defgammapl}
\end{equation}
where $\gammainc(\cdot,\frac{d_{ind}}{2})$ is the standard gamma incomplete function. Then following what was done in \cite{StojnicCSetam09,StojnicCSetamBlock09}
one has
\begin{equation}
\lim_{n\rightarrow\infty}E\hw_{\cweak}^+\approxeq\sqrt{2\gammaincpl^{-1}(\frac{1-\thetaweak}{1-\beta},d)},\label{eq:expcharcw}
\end{equation}
where $\gammaincpl^{-1}(\cdot,d)$ is the inverse of $\gammaincpl(\cdot,d)$.

\subsubsection{Explicit characterization of $\lim_{n\rightarrow\infty}\frac{E((\hw^+)^T\z)-E\sum_{i=1}^{\cweak}\hw_i^+}{n}$}
\label{sec:expcharsum}

One easily has
\begin{equation}
\lim_{n\rightarrow\infty}\frac{E((\hw^+)^T\z)-E\sum_{i=1}^{\cweak}\hw_i^+}{n}
=\lim_{n\rightarrow\infty}E\frac{\sum_{\cweak}^{n-\beta n}\hw_i^+}{n}=
\lim_{n\rightarrow\infty}E\frac{\sum_{(1-\thetaweak)n}^{(1-\beta) n}\hw_i^+}{n}.\label{eq:expchars1}
\end{equation}
Following further what was done in \cite{StojnicCSetam09,StojnicCSetamBlock09} one has
\begin{equation}
\lim_{n\rightarrow\infty}E\frac{\sum_{(1-\thetaweak)n}^{(1-\beta) n}\hw_i^+}{n}=(1-\beta)\sum_{d_{ind}=1}^{d}\frac{\binom{d}{d_{ind}}}{2^d}
\frac{\sqrt{2}\Gamma(\frac{d_{ind}+1}{2})}{\Gamma(\frac{d_{ind}}{2})}(1-\gammainc(\gammaincpl^{-1}(\frac{1-\thetaweak}{1-\beta},d),\frac{d_{ind}+1}{2})).\label{eq:expchars2}
\end{equation}
A combination of (\ref{eq:expchars1}) and (\ref{eq:expchars2}) gives
\begin{equation}
\lim_{n\rightarrow\infty}\frac{E((\hw^+)^T\z)-E\sum_{i=1}^{\cweak}\hw_i^+}{n}
=(1-\beta)\sum_{d_{ind}=1}^{d}\frac{\binom{d}{d_{ind}}}{2^d}
\frac{\sqrt{2}\Gamma(\frac{d_{ind}+1}{2})}{\Gamma(\frac{d_{ind}}{2})}(1-\gammainc(\gammaincpl^{-1}(\frac{1-\thetaweak}{1-\beta},d),\frac{d_{ind}+1}{2})).\label{eq:expchars3}
\end{equation}

\subsubsection{Explicit characterization of $\lim_{n\rightarrow\infty}\frac{E\sum_{\cweak}^{n+k}\hw_i^+}{n}$}
\label{sec:expcharsumsq}

We start with the following line of identities
\begin{eqnarray}
\frac{E\sum_{i=\cweak+1}^{n+k}\hw_i^2}{n} & = & \frac{E\sum_{i=\cweak+1}^{(1-\beta)n}\hw_i^2}{n}+\frac{E\sum_{i=(1-\beta)n+1}^n\hw_i^2}{n}
+\frac{E\sum_{i=n+1}^{n+\beta n}\hw_i^2}{n}\nonumber \\
& = & \frac{E\sum_{i=(1-\hthetaweak)n+1}^{(1-\beta)n}\hw_i^2}{n}+\frac{E\sum_{i=(1-\beta)n+1}^n\h_{(i-1)d+1}^2}{n}
+\frac{E\sum_{i=n+1}^{n+\beta n}\|\H_i^*\|_2^2}{n}\nonumber \\
& = & \frac{E\sum_{i=(1-\hthetaweak)n+1}^{(1-\beta)n}\hw_i^2}{n}+\frac{\beta n}{n} +\frac{\beta n(d-1)}{n}\nonumber \\
& = & \frac{E\sum_{i=(1-\hthetaweak)n+1}^{(1-\beta)n}\hw_i^2}{n}+\beta d.\label{eq:expcharsqs1}
\end{eqnarray}
Following further what was done in \cite{StojnicCSetam09,StojnicCSetamBlock09} one has
\begin{equation}
\lim_{n\rightarrow\infty}E\frac{\sum_{(1-\thetaweak)n}^{(1-\beta) n}(\hw_i^+)^2}{n}=(1-\beta)\sum_{d_{ind}=1}^{d}\frac{\binom{d}{d_{ind}}}{2^d}
\frac{2\Gamma(\frac{d_{ind}+2}{2})}{\Gamma(\frac{d_{ind}}{2})}(1-\gammainc(\gammaincpl^{-1}(\frac{1-\thetaweak}{1-\beta},d),\frac{d_{ind}+2}{2})).
\label{eq:expcharsqs2}
\end{equation}
A combination of (\ref{eq:expcharsqs1}) and (\ref{eq:expcharsqs2}) gives
\begin{equation}
\lim_{n\rightarrow\infty}\frac{E\sum_{i=\cweak+1}^{n+k}\hw_i^2}{n}
=(1-\beta) \sum_{d_{ind}=1}^{d}\frac{\binom{d}{d_{ind}}}{2^d}
\frac{2\Gamma(\frac{d_{ind}+2}{2})}{\Gamma(\frac{d_{ind}}{2})}(1-\gammainc(\gammaincpl^{-1}(\frac{1-\thetaweak}{1-\beta},d),\frac{d_{ind}+2}{2}))+\beta d.
\label{eq:expcharsqs3}
\end{equation}

We summarize the above results in the following theorem.

\begin{theorem}(Exact weak threshold)
Let $A$ be a $dm\times dn$ measurement matrix in (\ref{eq:system})
with the null-space uniformly distributed in the Grassmanian. Let
the unknown $\x$ in (\ref{eq:system}) be positive $k$-block-sparse with the length of its blocks $d$. Further, let the location and the directions of nonzero blocks of $\x$ be arbitrarily chosen but fixed.
Let $k,m,n$ be large
and let $\alpha=\frac{m}{n}$ and $\betaweak=\frac{k}{n}$ be constants
independent of $m$ and $n$. Let $\gammainc(\cdot,\cdot)$ and $\gammainc^{-1}(\cdot,\cdot)$ be the incomplete gamma function and its inverse, respectively. Further, let $\gammaincpl(\cdot,\cdot)$ be the following function
\begin{equation}
\gammaincpl(\cdot,d)=\sum_{d_{ind}=0}^d\frac{\binom{d}{d_{ind}}}{2^d}\gammainc(\cdot,\frac{d_{ind}}{2}),\label{eq:thmweakblockposdefgammapl}
\end{equation}
and let $\gammaincpl^{-1}(\cdot,\cdot)$ be its inverse. Let $\hthetaweak$, ($\betaweak\leq \hthetaweak\leq 1$) be the solution of
\begin{equation}
\frac{(1-\betaweak)\sum_{d_{ind}=1}^{d}\frac{\binom{d}{d_{ind}}}{2^d}
\frac{\sqrt{2}\Gamma(\frac{d_{ind}+1}{2})}{\Gamma(\frac{d_{ind}}{2})}(1-\gammainc(\gammaincpl^{-1}(\frac{1-\thetaweak}{1-\betaweak},d),\frac{d_{ind}+1}{2}))}{\hthetaweak}
\approxeq \sqrt{2\gammaincpl^{-1}(\frac{1-\thetaweak}{1-\betaweak},d)}.\label{eq:thmweakblockcond}
\end{equation}
\begin{enumerate}
\item If $\alpha$ and $\betaweak$ further satisfy
\begin{multline}
\alpha d>(1-\betaweak) \sum_{d_{ind}=1}^{d}\frac{\binom{d}{d_{ind}}}{2^d}
\frac{2\Gamma(\frac{d_{ind}+2}{2})}{\Gamma(\frac{d_{ind}}{2})}(1-\gammainc(\gammaincpl^{-1}(\frac{1-\hthetaweak}{1-\betaweak},d),\frac{d_{ind}+2}{2}))+\beta d\\
-\frac{((1-\betaweak)\sum_{d_{ind}=1}^{d}\frac{\binom{d}{d_{ind}}}{2^d}
\frac{\sqrt{2}\Gamma(\frac{d_{ind}+1}{2})}{\Gamma(\frac{d_{ind}}{2})}(1-\gammainc(\gammaincpl^{-1}(\frac{1-\hthetaweak}{1-\betaweak},d),\frac{d_{ind}+1}{2})))^2}{\hthetaweak}
\label{eq:thmweakblockposalpha1}
\end{multline}
then with overwhelming probability the solution of (\ref{eq:l2l1non}) is the positive $k$-block-sparse $\x$ from (\ref{eq:system}).

\item Moreover, if $\alpha$ and $\betaweak$ are such that
\begin{multline}
\alpha d<(1-\betaweak) \sum_{d_{ind}=1}^{d}\frac{\binom{d}{d_{ind}}}{2^d}
\frac{2\Gamma(\frac{d_{ind}+2}{2})}{\Gamma(\frac{d_{ind}}{2})}(1-\gammainc(\gammaincpl^{-1}(\frac{1-\hthetaweak}{1-\betaweak},d),\frac{d_{ind}+2}{2}))+\beta d\\
-\frac{((1-\betaweak)\sum_{d_{ind}=1}^{d}\frac{\binom{d}{d_{ind}}}{2^d}
\frac{\sqrt{2}\Gamma(\frac{d_{ind}+1}{2})}{\Gamma(\frac{d_{ind}}{2})}(1-\gammainc(\gammaincpl^{-1}(\frac{1-\hthetaweak}{1-\betaweak},d),\frac{d_{ind}+1}{2})))^2}{\hthetaweak}
\label{eq:thmweakblockposalpha2}
\end{multline}
then with overwhelming probability there will be a positive $k$-block-sparse $\x$ (from a set of $\x$'s with fixed locations and directions of nonzero blocks) that satisfies (\ref{eq:system}) and is \textbf{not} the solution of (\ref{eq:l2l1non}).
\end{enumerate}
\label{thm:thmweakthrblockpos}
\end{theorem}
\begin{proof}
The first part follows from the discussion presented above. The second part follows from the considerations presented in \cite{StojnicUpperBlock10,StojnicUpper10,StojnicRegRndDlt10}.
\end{proof}

\noindent \textbf{Remark:} To make writing easier in the previous theorem we removed all $\epsilon$'s used in Theorem \ref{thm:thmweakthrblock} and typically used in \cite{StojnicCSetam09,StojnicCSetamBlock09,StojnicUpper10,StojnicUpperBlock10}.

The above theorem essentially settles typical behavior of the $\ell_2/\ell_1$ optimization from (\ref{eq:l2l1non}) when used for solving (\ref{eq:system}) or (\ref{eq:l0}) assuming that $\x$ is a priori known to be positive and block-sparse.

The results for the weak threshold obtained from the above theorem
are presented in Figure \ref{fig:weakblocksppos}. More precisely, on the left hand side of Figure \ref{fig:weakblocksppos} we present the results that can be obtained from Theorem \ref{thm:thmweakthrblockpos}. In addition to that we on the right hand side of Figure \ref{fig:weakblocksppos} present the results one can obtain using Theorem \ref{thm:thmweakthrblock}. As is expected, given that the positive case assumes a bit more knowledge about $\x$ the recovery abilities of an algorithm (namely, in this case the one given in (\ref{eq:l2l1non})) tailored for such a case are a bit higher.
\begin{figure}[htb]
\begin{minipage}[b]{0.5\linewidth}
\centering
\centerline{\epsfig{figure=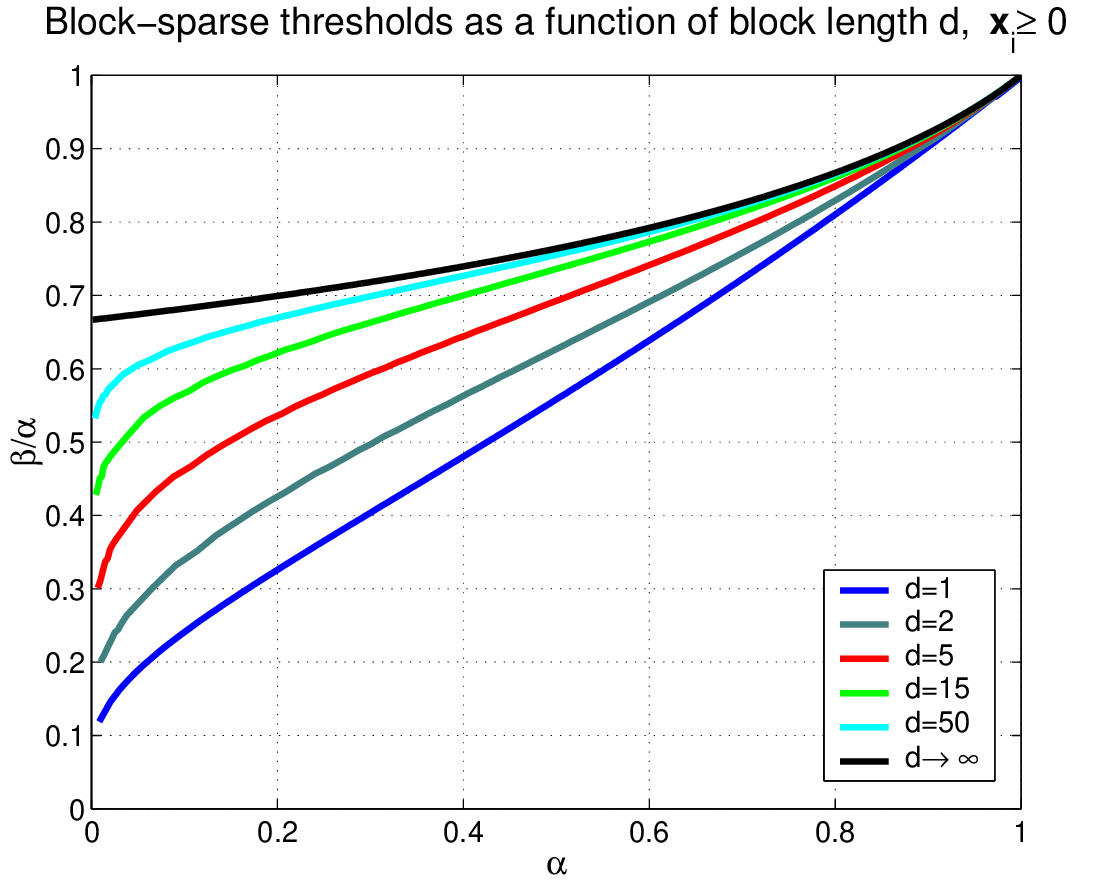,width=8cm,height=6cm}}
\end{minipage}
\begin{minipage}[b]{0.5\linewidth}
\centering
\centerline{\epsfig{figure=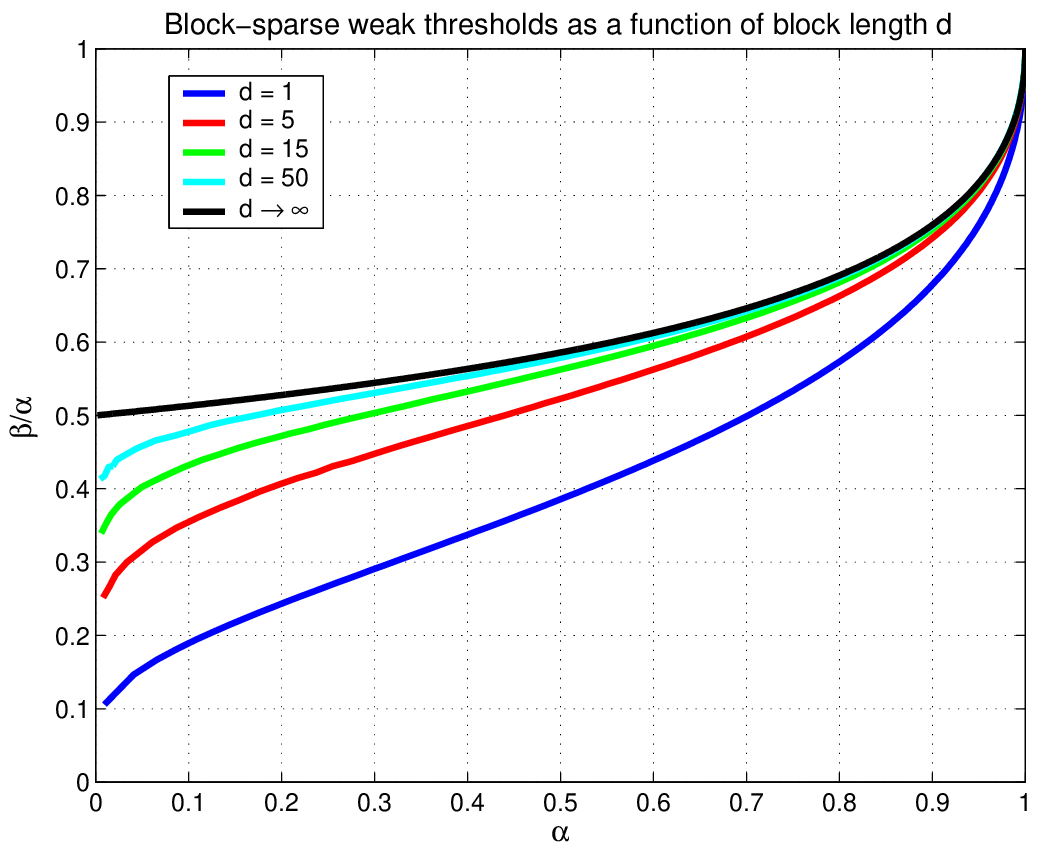,width=8cm,height=6cm}}
\end{minipage}
\caption{Left: Theoretical \emph{weak} threshold as a function of block length $d$ - $\x_i\geq 0,1\leq i\leq n$; Right: theoretical \emph{weak} threshold as a function of block length $d$ - general $\x$}
\label{fig:weakblocksppos}
\end{figure}

To get a feeling how accurately the presented analysis portraits the real behavior of the analyzed algorithms we below present a small set of results we obtained through numerical experiments.

\subsection{Numerical experiments}
\label{sec:simulp}

In this section we briefly discuss the results that we obtained from numerical experiments. In our numerical experiments we fixed $d=15$ and $n=100$ when $\alpha\geq 0.2$ On the other hand to get a bit of a finer resolution we set $n=150$ when $\alpha=0.1$. We then generated matrices $A$ of size $dm\times dn$ with $m=(15,20,30,\dots,90,99)$. The components of the measurement matrices $A$ were generated as i.i.d. zero-mean unit variance Gaussian random variables.
For each $m$ we generated randomly $k$-block-sparse positive signals $\x$ for several different values of $k$ from the transition zone (the locations of non-zero elements of $\x$ were chosen randomly as well). For each combination $(k,m)$ we generated $100$ different problem instances and recorded the number of times the  $\ell_2/\ell_1$-optimization algorithm from (\ref{eq:l2l1non}) failed to recover the correct $k$-block-sparse positive $\x$.
The obtained data are then interpolated and graphically presented on the right hand side of Figure \ref{fig:weakblocksppossimul}. The color of any point shows the probability of having $\ell_2/\ell_1$-optimization from (\ref{eq:l2l1non}) succeed for a combination $(\alpha,\beta)$ that corresponds to that point. The colors are mapped to probabilities according to the scale on the right hand side of the figure.
The simulated results can naturally be compared to the theoretical prediction from Theorem \ref{thm:thmweakthrblockpos}. Hence, we also show on the right hand side the theoretical value for the threshold calculated according to Theorem \ref{thm:thmweakthrblockpos} (and obviously shown on the left hand side of the figure as well). We observe that the simulation results are in a good agreement with the theoretical calculation.

\begin{figure}[htb]
\begin{minipage}[b]{0.5\linewidth}
\centering
\centerline{\epsfig{figure=WeakThrBlSpNon.eps,width=8cm,height=6cm}}
\end{minipage}
\begin{minipage}[b]{0.5\linewidth}
\centering
\centerline{\epsfig{figure=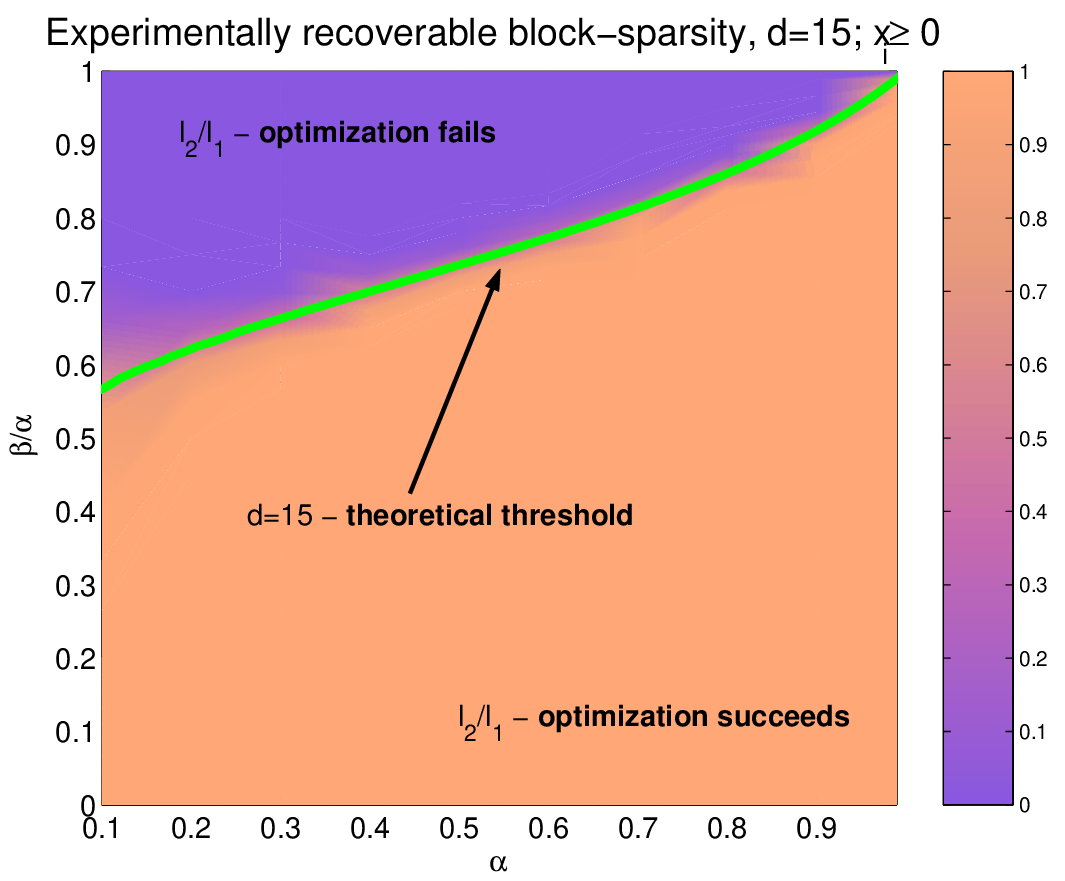,width=8cm,height=6cm}}
\end{minipage}
\caption{Left: Theoretical \emph{weak} threshold as a function of fraction of \emph{hidden} known support; Right: Experimentally recoverable sparsity; fraction of \emph{hidden} known support $\eta=\frac{3}{4}$}
\label{fig:weakblocksppossimul}
\end{figure}

\subsection{$d\rightarrow\infty$ -- weak threshold}
\label{sec:dinftyweak}

When the block length is large one can simplify the conditions for finding the thresholds obtained in the previous section. Hence,
in this section we establish weak thresholds when $d\rightarrow\infty$, i.e. we establish ultimate benefit of $\ell_2/\ell_1$-optimization from (\ref{eq:l2l1non}) when used in recovery of block-sparse positive vectors from (\ref{eq:system}). Throughout this section we choose $d\rightarrow\infty$ in order to simplify the exposition.

Following the reasoning presented in \cite{StojnicCSetamBlock09} it is not that difficult to see that choosing $\htheta_w=1$ in (\ref{eq:thmweakblockposalpha1}) would provide a valid threshold condition as well (in general $\htheta_w=1$ is not optimal for a fixed value $d$, i.e. when $d$ is not large a better choice for $\htheta_w$ is the one given in Theorem \ref{thm:thmweakthrblockpos}). However it can be shown that the choice $\htheta_w=1$ gives us the following corollary of Theorem \ref{thm:thmweakthrblockpos}.

\begin{corollary}($d\rightarrow\infty$)
Let $A$ be a $dm\times dn$ measurement matrix in (\ref{eq:system})
with the null-space uniformly distributed in the Grassmanian. Let
the unknown $\x$ in (\ref{eq:system}) be positive $k$-block-sparse with the length of its blocks $d\rightarrow\infty$. Further, let the location and the directions of nonzero blocks of $\x$ be arbitrarily chosen but fixed.
Let $k,m,n$ be large
and let $\alpha=\frac{m}{n}$ and $\betaweak^{\infty}=\frac{k}{n}$ be constants
independent of $m$ and $n$. Assume that $d$ is independent of $n$.
If $\alpha$ and $\betaweak^{\infty}$ satisfy
\begin{equation}
\alpha > \frac{\betaweak^{\infty}(3-\betaweak^{\infty})}{2}\label{eq:corweakalpha}
\end{equation}
then the solution of (\ref{eq:l2l1non}) is with overwhelming
probability the positive $k$-block sparse $\x$ in (\ref{eq:system}).\label{thm:corweakthr}
\end{corollary}
\begin{proof}
Let $\hthetaweak\rightarrow 1$ in (\ref{eq:thmweakblockposalpha1}). Then from (\ref{eq:thmweakblockposalpha1}) we have
\begin{eqnarray}
\alpha & > & \frac{(1-\betaweak)d/2
+\betaweak d}{d} - \frac{\left ((1-\betaweak)\frac{\sqrt{2}\Gamma(\frac{d/2+1}{2})}{\Gamma(\frac{d}{4})}\right ) ^2}{d}\nonumber \\
& = & \frac{1+\betaweak}{2}- \frac{\left ((1-\betaweak)\frac{\sqrt{2}\Gamma(\frac{d/2+1}{2})}{\Gamma(\frac{d}{4})}\right ) ^2}{d}.\label{eq:corweakalpha1}
\end{eqnarray}
When $d\rightarrow \infty$ we have $\lim_{d\rightarrow \infty}\frac{1}{d}\left (\frac{\sqrt{2}\Gamma(\frac{d/2+1}{2})}{\Gamma(\frac{d}{4})}\right )^2=\frac{1}{2}$. Then from (\ref{eq:corweakalpha1}) we easily obtain the condition
\begin{equation*}
\alpha > \frac{\betaweak(3-\betaweak)}{2}\label{eq:corweakalpha2}
\end{equation*}
which is the same as the condition stated in (\ref{eq:corweakalpha}). This therefore concludes the proof.
\end{proof}
The results for the weak threshold obtained in the above corollary are shown in figures in earlier sections as curves denoted by $d\rightarrow\infty$.

\section{Conclusion}
\label{sec:conc}

In this paper we studied a variant of the standard compressed sensing setup. The variant that we studied assumes vectors that are sparse but also structured. The type of structure that we studied is the co-called block-sparsity. While the standard block-sparsity has been studied in \cite{StojnicCSetamBlock09,StojnicUpperBlock10} here we combine it with another type of structure, that accounts for a priori known (same) signs of unknown vectors.

Typically when the unknown vectors are block-sparse one handles them by employing an $\ell_2/\ell_1$ norm combination in place of the standard $\ell_1$ norm. We looked at a signed modification of $\ell_2/\ell_1$ norm and analyzed how it fares when used for recovery of block-sparse signed vectors from an under-determined system of linear equations. The analysis we provided viewed linear systems in a statistical context. For such systems we then established a precise probabilistic characterization of problem dimensions for which the signed modification of $\ell_2/\ell_1$ norm is guaranteed to recover with overwhelming probability the block-sparsest positive unknown vector.

As was the case with many results we developed (see, e.g. \cite{StojnicCSetam09,StojnicMoreSophHopBnds10,StojnicLiftStrSec13,StojnicTowBettCompSens13}), the purely theoretical results we presented in this paper are valid for the so-called Gaussian models, i.e. for systems with i.i.d. Gaussian coefficients. Such an assumption significantly simplified our exposition. However, all results that we presented can easily be extended to the case of many other models of randomness. There are many ways how this can be done. Instead of recalling on them here we refer to a brief discussion about it that we presented in \cite{StojnicMoreSophHopBnds10,StojnicLiftStrSec13}.

As for usefulness of the presented results, similarly to many of the results we created within compressed sensing, there is hardly any limit. One can look at a host of related problems from the compressed sensing literature. These include for example, all noisy variations, approximately sparse unknown vectors, vectors with a priori known structure (block-sparse, binary/box constrained etc.), all types of low rank matrix recoveries, various other algorithms like $\ell_q$-optimization, SOCP's, LASSO's, and many, many others. Each of these problems has its own specificities and adapting the methodology presented here usually takes a bit of work but in our view is now a routine. While we will present some of these applications we should emphasize that their contribution will be purely on an application level.

\begin{singlespace}
\bibliographystyle{plain}
\bibliography{PosBlSpRefs}
\end{singlespace}

\end{document}